\documentclass{article}

\usepackage{amsmath,amsthm,amssymb,mathtools}
\usepackage{graphicx}
\usepackage{booktabs}
\usepackage[numbers]{natbib}
\usepackage{hyperref}
\usepackage[capitalise]{cleveref}
\crefname{equation}{Eq.}{Eqs.}
\Crefname{equation}{Eq.}{Eqs.}

\theoremstyle{plain}
\newtheorem{lemma}{Lemma}[section]
\theoremstyle{remark}

\renewcommand{\Re}{\operatorname{Re}}
\newcommand{\E}{\mathbb{E}}
\newcommand{\C}{\mathbb{C}}
\newcommand{\Strig}{S_{\mathrm{trig}}}
\newcommand{\Snorm}{S_{\mathrm{norm}}}
\newcommand{\abs}[1]{\lVert #1 \rVert}

\title{\large Precomputing the Future-Offset Average in TriAttention}
\author{%
  \normalsize Amarnath Mukherjee\thanks{Correspondence: \texttt{am [at] hozhoke [dot] com}} \\
  \normalsize Hozhoke, Inc.
}
\date{\normalsize June 10, 2026 \\ Last revised: June 29, 2026}

\begin{document}
\maketitle

\begin{abstract}
TriAttention is a recent method for shrinking the KV cache of long-reasoning
LLMs: it scores each cached key by how much attention it is likely to receive and
evicts the lowest-scoring ones. Because a key does not know how far away its
future queries will sit, the score is averaged over a ladder of $17$ possible
future distances. We point out that this average is free: the future distance
enters the score only through the position-dependent rotation, so the whole
$17$-fold average collapses---exactly, by a one-line algebraic identity---into a
single per-band weight that is computed once, offline. Scoring a key then costs
one evaluation instead of seventeen, with no change to which keys get pruned. The
saving is modest and lives entirely in TriAttention's pruning-score computation,
not in the attention kernel; we present it as a small, exact complement to their
method, and we confirm the identity numerically.
\end{abstract}

\section{Introduction}
\label{sec:intro}

In long-context LLMs the KV cache grows large, and to stay within a memory budget
it must periodically be pruned. Which entries to evict is a consequential choice:
drop a key the model still needs and its accuracy suffers.

TriAttention~\citep{triattention2026} is a recent and elegant answer to that
choice. When the cache must be trimmed, it scores each cached key by how much
attention it is likely to receive, and keeps the highest-scoring ones.

This paper contributes a small improvement to one step of that method. We first
describe the part of TriAttention our result touches, then the improvement itself.

\section{The TriAttention score}
\label{sec:setup}

TriAttention's score rests on a striking regularity in how trained models
represent queries and keys.

\paragraph{Their empirical observation.}
Across several leading-edge models (Qwen3, Qwen2.5, Llama3), and measured
\emph{before} the RoPE position-dependent rotation~\citep{rope} is applied, a head's
query and key vectors do not spread out. They cluster tightly around a fixed,
non-zero center, and that center barely moves across positions or inputs. Roughly
$90\%$ of heads show this, and it is stable across task domains. The consequence
matters asymmetrically for what comes next. At pruning time a cached key is
already known exactly; the only thing that must be approximated is the
\emph{future} query that will read it---and because queries concentrate, their
center is a faithful stand-in for it. With the query pinned to that center, the
only thing still varying in the attention logit is the rotation, i.e.\ the
query--key distance $\Delta$. Attention as a function of distance becomes a fixed,
predictable curve you can compute once, offline---no need to watch live attention
scores. This is why TriAttention scores keys in the pre-rotation space: it is where
the query center is stable enough to be a reliable summary.

\paragraph{TriAttention Score.}
At pruning time the keys are concrete (they sit in the cache), but the queries
that will read them are in the future and do not exist yet. So TriAttention
substitutes the calibrated query center $\E[q_f]$---the average of band-$f$
queries over a calibration set---for the unknown future query and
predicts the attention a key $k$ would receive at distance $\Delta$
(\citealp{triattention2026}, Eq.~10):
\begin{equation}
  S(k,\Delta)=\Strig(k,\Delta)+\Snorm(k),
  \label{eq:score}
\end{equation}
whose trigonometric term (\citealp{triattention2026}, Eq.~6) collects the
per-band contributions,
\begin{equation}
  \Strig(k,\Delta)=\sum_f \abs{\E[q_f]}\,\abs{k_f}\cos(\omega_f\Delta+\varphi_f),
  \label{eq:strig}
\end{equation}
where:
\begin{itemize}
  \item $f\in\{0,\dots,d/2-1\}$ indexes the rotation's two-dimensional bands,
        which TriAttention calls \emph{frequency bands};
  \item $\omega_f=\theta^{-2f/d}$ is band $f$'s \emph{angular} frequency---the
        rate at which the rotation turns that band;
  \item $\varphi_f=\arg\E[q_f]-\arg k_f$ is the angle between the query center
        and the key's band-$f$ component $k_f$;
  \item $d$ is the head dimension; and
  \item $\theta=10000$ is the rotary base TriAttention uses---the value RoPE
        adopted from the sinusoidal positional encoding of the original
        Transformer~\citep{vaswani2017}.
\end{itemize}
The second term $\Snorm(k)$ is a norm-based correction that
does not depend on $\Delta$; it is mentioned for completeness,
but is not relevant to our discussion in this paper.

\begin{figure}[t]
  \centering
  \includegraphics[width=\textwidth]{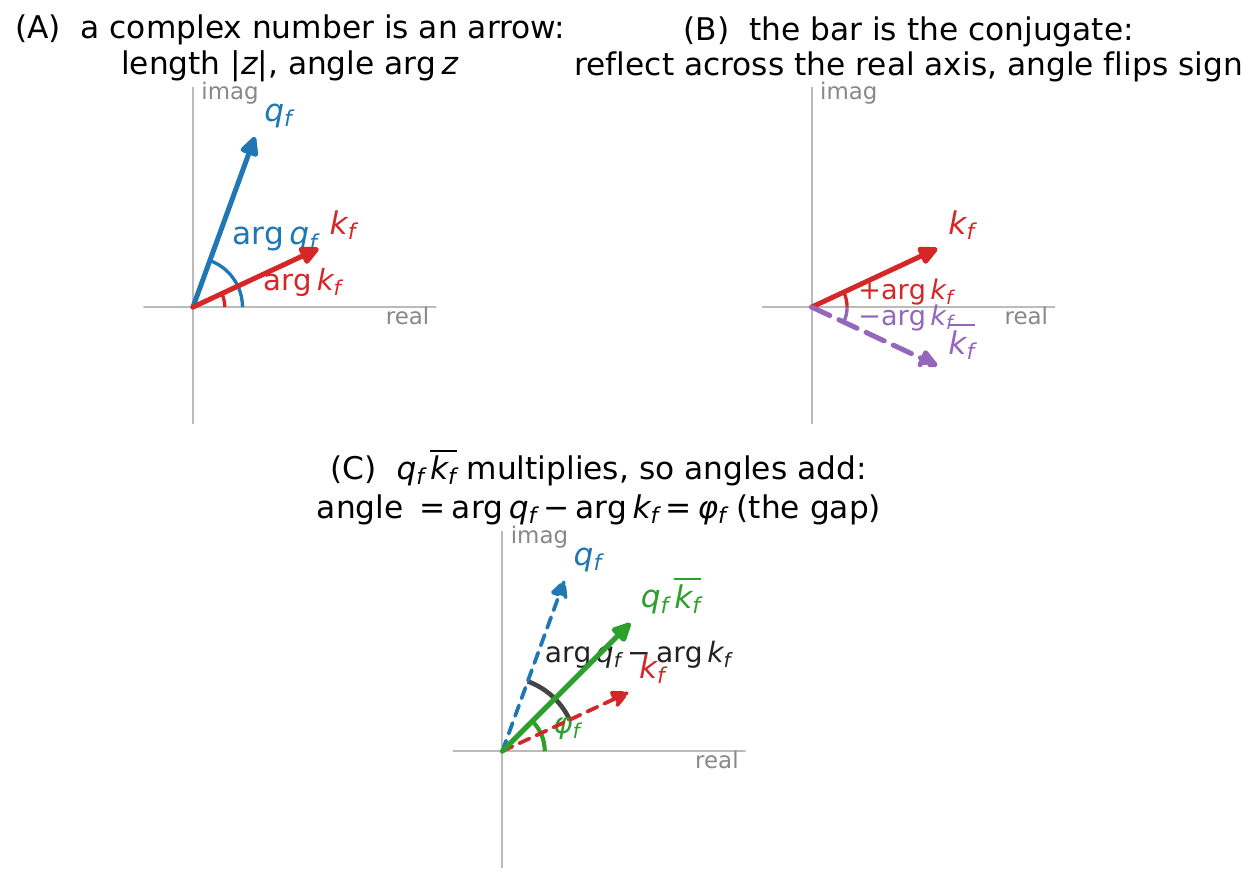}
  \vspace{8pt}
  \caption{Reading band $f$ as arrows in the plane. \textbf{(A)}~A complex number
  is an arrow with a length and an angle; $q_f$ and $k_f$ are the pre-rotation
  query and key in this band. \textbf{(B)}~The conjugate $\overline{k_f}$ is the
  mirror image of $k_f$ across the real axis---it has the same length, but its angle flips sign.
  \textbf{(C)}~Because multiplying complex numbers adds their angles,
  $q_f\,\overline{k_f}$ is one arrow carrying the angle \emph{between} query and
  key, $\arg q_f-\arg k_f$. Multiplying further by
  $e^{\,i\omega_f\Delta}$ turns it by the distance angle, and its real part is the
  cosine term of the score.}
  \label{fig:phase}
\end{figure}

\paragraph{The score as a picture in the plane.}

The rotation acts on two coordinates of the query and key at a time, turning that
pair by an angle. A pair of coordinates is a point in the plane, and a point in
the plane is a complex number: an arrow with a length and an angle. So read band
$f$'s pre-rotation query and key as complex numbers $q_f,k_f\in\C$
(\cref{fig:phase}A). Three small facts are all we need.

\begin{enumerate}

\item The bar flips the angle. The conjugate $\overline{k_f}$ is the
mirror image of $k_f$ across the real axis: same length, opposite angle
(\cref{fig:phase}B).
\item Multiplying adds angles. So the product
$q_f\,\overline{k_f}$ is a single arrow whose angle is the query--key gap
$\arg q_f-\arg k_f$, and whose length is $\abs{q_f}\abs{k_f}$ (\cref{fig:phase}C);
with the center in the query's place this gap is $\varphi_f$ (\cref{eq:strig}).
\item The
position-dependent rotation is itself just a multiplication: turning band $f$ by
the distance angle $\omega_f\Delta$ means multiplying by $e^{\,i\omega_f\Delta}$.

\end{enumerate}

Put the three together, now with the score's center $\E[q_f]$ in place of the
unknown query. The arrow $\E[q_f]\,\overline{k_f}\,e^{\,i\omega_f\Delta}$ has length
$\abs{\E[q_f]}\abs{k_f}$ and angle $\varphi_f+\omega_f\Delta$, and its real part is
exactly $\abs{\E[q_f]}\abs{k_f}\cos(\omega_f\Delta+\varphi_f)$---one term of
$\Strig$. Summing over bands, the trigonometric part of the score,
\cref{eq:strig}, is the real part of a sum of arrows---the same complex form
RoFormer gives the rotary dot product (\citealp[Eq.~12]{rope}):
\begin{equation}
  \Strig(k,\Delta)=\Re\!\Big\{\sum_f \E[q_f]\,\overline{k_f}\,e^{\,i\omega_f\Delta}\Big\}.
  \label{eq:phasor}
\end{equation}

\section{The observation}
\label{sec:result}

\paragraph{The averaging step---the one we improve.}
A cached key might be queried soon or much later. Its real importance is
its importance across all those future distances, so TriAttention averages the
score over a set of future offsets
$D=\{1,2,4,\dots,2^{16}\}$ (\citealp{triattention2026}, Eq.~11):
\begin{equation}
  \tilde S(k)=\frac{1}{|D|}\sum_{\delta\in D} S(k,\Delta+\delta),
  \qquad |D|=17.
  \label{eq:avg}
\end{equation}
The offsets are spaced geometrically (powers of two) because attention changes
fast at short distances and slowly at long ones. Done as written,
\cref{eq:avg} evaluates the whole band sum $17$ times per key and averages---and
that is exactly what the reference implementation does, materializing a
[keys $\times$ offsets $\times$ bands] table of cosines and averaging over the
offset axis.\footnote{Reference implementation~\citep{triattentioncode}:
\texttt{triattention/methods/pruning\_utils.py::score\_keys\_for\_round}, the
default mean-aggregation path.}

Now, consider where the offset $\delta$ goes. In the phasor score, \cref{eq:phasor},
it appears only inside the rotation, and---because adding angles multiplies
rotations---a turn by the shifted distance factors into the part we keep and the
part we average:
\begin{equation}
  e^{\,i\omega_f(\Delta+\delta)}
   =\underbrace{e^{\,i\omega_f\Delta}}_{\text{distance, kept}}\;
    \underbrace{e^{\,i\omega_f\delta}}_{\text{offset, averaged}}.
  \label{eq:split}
\end{equation}
Only the offset factor depends on $\delta$, and it touches neither the key nor the
center. So averaging the phasor score over the $17$ offsets
reaches only that factor, where it collapses to one complex weight per band:
\begin{equation}
  \tilde S(k)=\Re\!\Big\{\sum_f \E[q_f]\,\overline{k_f}\,e^{\,i\omega_f\Delta}\,
   \underbrace{\tfrac{1}{|D|}\sum_{\delta\in D}e^{\,i\omega_f\delta}}_{\textstyle W_f}\Big\}+\Snorm(k).
  \label{eq:collapse}
\end{equation}
The averaged offset rotation $W_f=\frac{1}{|D|}\sum_{\delta\in D}e^{\,i\omega_f\delta}$
(\cref{fig:Wf}) depends only on the frequencies
$\{\omega_f\}$ and the fixed offset set, so it is computed once, offline, alongside
the centers.

With $W_f$ in hand, the averaged score is a \emph{single} evaluation of the original
score, \cref{eq:score}, with the query center $\E[q_f]$ nudged to a precomputed
``offset-averaged center''\footnote{If a band has $\abs{\mu_f'}=0$ or
$\abs{k_f}=0$, its phase is immaterial---the band contributes zero, and the phasor
form~\eqref{eq:collapse} is the canonical statement.} $\mu_f'=\E[q_f]\,W_f$, with
$\varphi_f'=\arg\mu_f'-\arg k_f$:
\begin{equation}
  \boxed{\;\tilde S(k)=\sum_f \abs{\mu_f'}\,\abs{k_f}\cos(\omega_f\Delta+\varphi_f')+\Snorm(k),
  \qquad \mu_f'=\E[q_f]\,W_f.\;}
  \label{eq:folded}
\end{equation}
After that, scoring a key costs one band-sum evaluation instead of $17$,
and---because \cref{eq:folded} is an exact rewrite of \cref{eq:avg}, not an
approximation---every key receives precisely the score it had before. The pruned
set is unchanged. The complex weights $W_f$ need not appear in the scoring loop:
folded offline into the centers (or the real coefficients), they leave the
computation on the real cosine/sine form TriAttention already uses;
\cref{app:real} makes that concrete as a fixed $2\times2$ per-band map and
reports a numerical check.

\begin{figure}[t]
  \centering
  \includegraphics[width=0.92\textwidth]{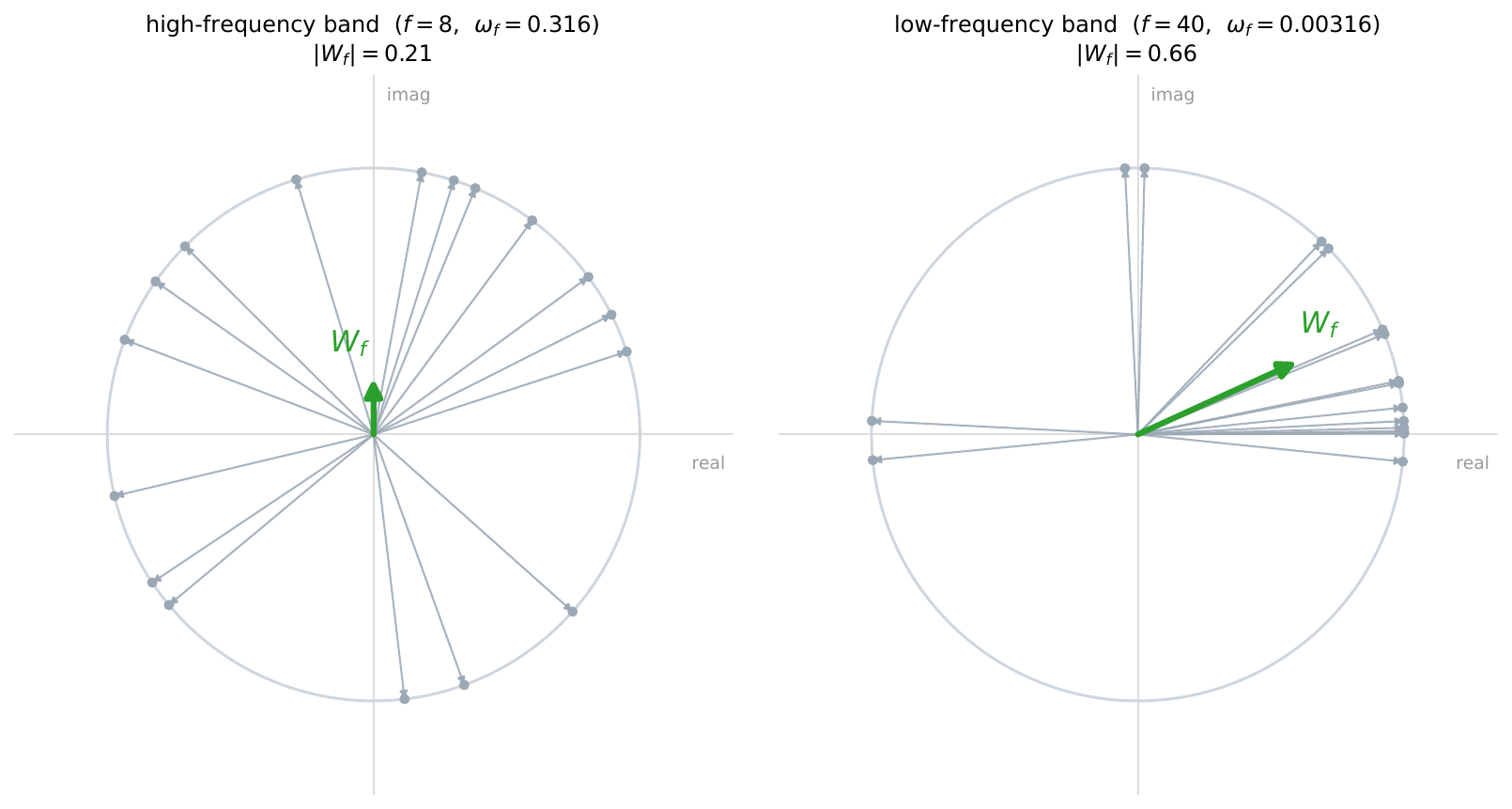}
  \caption{The $17$-fold offset average is an average of rotations. Each offset
  $\delta\in D$ contributes a unit arrow $e^{\,i\omega_f\delta}$ on the circle;
  their average is the single weight $W_f$ (green). In a high-frequency band
  (left) the offsets fan out and largely cancel, so $W_f$ is short; in a
  low-frequency band (right) they cluster and reinforce, so $W_f$ is long. Either
  way it is one number per band, computed once.}
  \label{fig:Wf}
\end{figure}

\section{What it saves, and what it does not}
\label{sec:cost}

For one pruning event over $N$ cached keys and $F$ bands, the trigonometric term
drops from $O(N F\,|D|)$ band-cosine evaluations to $O(N F)$, plus a one-time
offline cost of $O(F\,|D|)$ to build the weights $\{W_f\}$. The reduction on the first
term is an exact factor of $|D|=17$.

It is worth being clear about the size and place of this. The saving is in the
\emph{scoring arithmetic of the pruning step}, which TriAttention already runs
only once per window of $\beta=128$ generated tokens. It is not a saving in the
attention kernel that runs every token, nor in the cache budget, nor in how often
pruning happens. It composes with, and does not replace, the method's other
engineering: the window batching that decides how \emph{often} scoring runs, and
the grouped-query normalize-then-max step that consumes the per-head
$\tilde S(k)$ unchanged. This rewrite is exact for the average in \cref{eq:avg}, which
is the published method and the default in the code; it relies on the average
being linear, so it does \emph{not} apply to a non-default variant that takes a
maximum over offsets instead of a mean. In short: a key step that looked like it
cost $17\times$ costs $1\times$, paid for once, offline, without changing any output---a
small, exact tightening of one bolt in their machine.

\section{Conclusion}
\label{sec:conclusion}

TriAttention's future-offset average admits an offline computation: because the offset lives
entirely in the position-dependent rotation, it folds exactly into a per-band
weight that is precomputed offline, turning a $17$-fold loop into a single
evaluation, with no change to which keys are kept. We offer it as a small, exact
complement to a method we expect to see more of, and as a reminder that ``average
over a geometric ladder of rotations'' is often a precompute in disguise.

\appendix
\crefalias{section}{appendix}  
\section{Derivation}
\label{app:derivation}

\setlength{\abovedisplayskip}{14pt plus 3pt minus 2pt}
\setlength{\belowdisplayskip}{14pt plus 3pt minus 2pt}
\setlength{\abovedisplayshortskip}{8pt plus 3pt}
\setlength{\belowdisplayshortskip}{12pt plus 3pt minus 2pt}

The rewrite used in the body is exact; here is the one line behind it.

\begin{lemma}[Offset separability]
\label{lem:fold}
For the mean future-offset score in \cref{eq:avg}, the collapse in \cref{eq:collapse}---and
its polar form in \cref{eq:folded}, with $\mu_f'=\E[q_f]\,W_f$ and
$\varphi_f'=\arg\mu_f'-\arg k_f$---holds exactly.
\end{lemma}

\begin{proof}
Substitute the phasor score in \cref{eq:phasor} into the average in \cref{eq:avg},
and use linearity of $\Re$ and of the finite band sum:
\[
  \frac{1}{|D|}\sum_{\delta\in D}\Strig(k,\Delta+\delta)
  =\Re\!\Big\{\sum_f \E[q_f]\,\overline{k_f}\,e^{\,i\omega_f\Delta}\,
   \underbrace{\tfrac{1}{|D|}\textstyle\sum_{\delta\in D}e^{\,i\omega_f\delta}}_{W_f}\Big\}.
\]
The norm term $\Snorm(k)$ is independent of $\delta$, so the average leaves it
unchanged; the polar form follows from the polar decomposition of
$\mu_f'\overline{k_f}$.
\end{proof}

\section{Real-coefficient form}
\label{app:real}

\begingroup\linespread{1.12}\selectfont
No complex arithmetic is required at inference. TriAttention writes the
single-distance score with constant real coefficients
(\citealp{triattention2026}, Eqs.~17--19),
\[
  \Strig(k,\Delta)=\sum_f\big[a_f\cos(\omega_f\Delta)+b_f\sin(\omega_f\Delta)\big],
\]
with $a_f=\abs{\E[q_f]}\abs{k_f}\cos\varphi_f$ and
$b_f=-\abs{\E[q_f]}\abs{k_f}\sin\varphi_f$. In these coordinates,
\cref{eq:collapse} is a fixed $2\times2$ per-band linear map, with
$A_f=\Re W_f$ and $B_f=\Im W_f$:
\begin{equation}
\begin{aligned}
  \begin{pmatrix}a_f'\\ b_f'\end{pmatrix}
  &=\begin{pmatrix}A_f & B_f\\ -B_f & A_f\end{pmatrix}
   \begin{pmatrix}a_f\\ b_f\end{pmatrix},\\[4pt]
  \tilde\Strig(k)&=\sum_f\big[a_f'\cos(\omega_f\Delta)+b_f'\sin(\omega_f\Delta)\big].
\end{aligned}
  \label{eq:realmap}
\end{equation}
\leavevmode\\[9pt]
The existing real-valued scoring form runs unchanged on the pre-mapped coefficients;
complex numbers appear only in this derivation. A per-band, offset-independent
rescaling that the code already applies (a high-frequency mask) multiplies
$\abs{\mu_f'}$ and is unaffected.
\par\endgroup

\paragraph{Numerical confirmation.}
We checked \cref{eq:folded,eq:realmap} by comparing the literal $17$-fold offset
loop (a faithful reimplementation of the reference scorer) against the
complex, polar, and real-coefficient forms on synthetic centers and keys. All three
match the baseline to machine precision ($\max|\Delta|\!\sim\!10^{-12}$) across
$\delta_{\max}\in\{2^7,2^{12},2^{16}\}$.

\bibliographystyle{plainnat}
\bibliography{references}

\end{document}